\documentclass[runningheads]{llncs}

\usepackage[T1]{fontenc}
\usepackage[utf8]{inputenc}
\usepackage{amssymb, bm}
\usepackage{braket}
\usepackage{graphicx}
\usepackage{hyperref}
\usepackage{color}
\usepackage{physics}
\usepackage{booktabs}
\usepackage{enumitem}

\newcommand{\Hil}{\mathcal{H}}
\newcommand{\Real}{\mathbb{R}}
\newcommand{\Complex}{\mathbb{C}}

\begin{document}

\title{The Inverse Born Rule Equivalence.\\[4pt]
\large On the Informational Limits of Real-Valued Amplitude Encodings
and the Measurement of Quantum Advantage in Data Embeddings}

\titlerunning{The Inverse Born Rule Equivalence}

\author{Sebastian Zaj\k{a}c\inst{1} \and
Jacob L. Cybulski\inst{2,3} \and
Bartosz Dziewit\inst{4} \and
Tomasz Kulpa\inst{5}}

\authorrunning{S. Zaj\k{a}c et al.}

\institute{Military University of Technology, Warsaw, Poland\\
\email{sebastian.zajac@wat.edu.pl}\and
Enquanted, Australia\and
Deakin University, Melbourne, Australia\and
Institute of Physics, University of Silesia in Katowice, Poland\and
Cardinal Stefan Wyszynski University in Warsaw, Poland}

\maketitle

\begin{abstract}
When does quantum data encoding provide genuine quantum advantage, 
and when does it merely rephrase a classically solvable problem? 

We prove an \emph{Equivalence Theorem} demonstrating that any encoding mapping classical data to real-valued amplitudes,  $\ket{\psi_c} = \sum_i c_i \ket{i}$ with $c_i \in \Real$ and $\sum_i c_i^2 = 1$, composed with a data-independent parameterised unitary and computational-basis measurement, yields exactly the class of classical quadratic forms. 
We identify the geometric mechanism driving this collapse: the restriction to $\Real$ forces a vanishing Berry connection, removing the complex phases required for data-dependent quantum interference. 
To operationalize this boundary, we introduce encoding diagnostics---phase complexity $C[\Phi]$ and mode-wise von Neumann mutual information $I[\Phi]$---and link them to the information-geometric excess $\Delta g$. 
We show that for all real-valued 
encodings, $\Delta g = 0$ identically.
We term the misidentification of such models as evidence of quantum computational power the \emph{Inverse Born Rule Fallacy}. Supported by numerical experiments, our results establish that complex-phase structure is a strictly necessary condition for data-driven (Type~B) quantum advantage.
\end{abstract}

\section{Introduction}
Quantum Machine Learning (QML) seeks to exploit quantum interference to reach a class of hypotheses inaccessible to classical models.
A central design decision in any QML pipeline is the \emph{data encoding}:
the quantum feature map $\Phi : \mathcal{X} \to \mathcal{D}(\Hil)$ that embeds
classical data into quantum states.
The encoding determines which functions the model can express; yet, despite
growing theoretical and empirical interest
\cite{Schuld2019,Havlicek2019,PerezSalinas2020}, there is no general
criterion for when an encoding provides a genuine quantum advantage.
As we show, a widely used family of encodings provably does not.

\subsection{Two Types of Quantum Advantage}




We distinguish two paradigms of ``quantum advantage'' that are frequently 
conflated in the QML literature.

\emph{Type~A advantage} refers to pure computational speedups derived from 
the inherent hardness of simulating a quantum circuit, independent of 
the data encoding. Canonical examples include Grover's search \cite{Grover1996} 
and Shor's factoring \cite{Shor1994}.

\emph{Type~B advantage}, in contrast, refers to an enlargement of the 
hypothesis class. Here, the set of input--output functions expressible 
by the model is strictly richer than what any efficient classical model 
can achieve. Crucially, this enrichment depends entirely on the 
encoding~$\Phi$.

These two paradigms are distinct. For instance, Grover's algorithm 
utilizes amplitude encoding but yields only a Type~A advantage; its 
speedup stems from the oracle's structure, rather than the statistical 
properties of the encoded data.\begingroup
\renewcommand{\thefootnote}{*}\footnote{Our results therefore do not contradict the efficiency of Grover's algorithm or any other Type~A speedup.}\endgroup

This paper focuses exclusively on Type~B advantage. Our central result 
demonstrates that real-valued amplitude encoding---including probability 
loading as a non-negative sub-case---provides zero Type~B advantage. 
Regardless of the subsequent parameterised ansatz, its hypothesis class 
remains entirely classically realisable via a standard inner-product kernel.

\subsection{Contributions}

We make four primary contributions:%
\begin{itemize}[topsep=0pt]
    \item \textbf{Equivalence Theorem (Section~\ref{sec:equiv}):} We prove that real-valued amplitude encoding---including non-negative probability loading as a special case---composed with any data-independent unitary yields a hypothesis class that reduces exactly to a classically realisable quadratic form.
    
    \item \textbf{Geometric Mechanism (Section~\ref{sec:geometry}):} We show that restricting amplitudes to~$\Real$ forces the Berry connection to vanish. This restriction collapses the quantum Fubini--Study metric directly to the classical Fisher--Rao metric, eliminating the geometry required for quantum advantage.
    
    \item \textbf{Encoding Diagnostics (Section~\ref{sec:quantumness}):} We introduce three quantitative metrics to evaluate the nature of quantum encoding: phase complexity~$C[\Phi]$, the magnitude of the Berry connection~$|A_\mu|$, and mode-wise mutual information~$I[\Phi]$---and formally link them with information-geometric excess~$\Delta g$.
    
    \item \textbf{Numerical Validation (Section~\ref{sec:experiments}):} We empirically validate our theoretical boundary using nonlinear classification tasks in five different encoding strategies, demonstrating that our diagnostics reliably predict where classical baselines plateau.
\end{itemize}

\section{Related Work}
\label{sec:related}

Schuld and Killoran~\cite{Schuld2019} established that variational
quantum classifiers are equivalent to kernel methods, with the quantum
feature map defining an implicit kernel
$k(x,x') = |\!\braket{\psi(x)}{\psi(x')}\!|^2$.
That result guarantees the existence of a kernel representation but leaves open which kernels are classically tractable and wyhy.
We answer both questions for real-valued amplitude encodings: the induced kernel is the Hellinger affinity $k_H$, computable in $\mathcal{O}(N^2)$ without quantum hardware (Theorem~1), and the mechanism is a vanishing Berry connection that makes the quantum state manifold isometric to the classical Fisher–Rao geometry (Proposition~1).
Havl\'\i\v{c}ek et al.~\cite{Havlicek2019} showed that such kernels
can be hard to estimate classically when the feature map involves deep
entangling circuits.
The expressive power of the model therefore depends on the
choice of~$\Phi$---a dependence this paper makes precise for
real-valued amplitude encodings.

Among encoding strategies, data re-uploading~\cite{PerezSalinas2020}
demonstrated that repeated data injection enables single-qubit circuits
to approximate any continuous function, suggesting that the manner in
which data participates in the circuit---not merely the initial
embedding---is decisive for expressivity.

Independently, Zhao et al.~\cite{Zhao2026} identified a $\sqrt{p}$
bottleneck in quantum oracle sketching: the precision with which a
probability distribution can be queried via amplitude oracles is
fundamentally limited by the $\ell_2$ structure of~$\sqrt{p}$.
Their result converges with our characterisation of real-valued
amplitude encoding through the Born rule lens, though they address
pure computational complexity (Type~A) rather than hypothesis class 
expressivity (Type~B).

A classical result in information geometry, dating to
\v{C}encov~\cite{Cencov1982}, establishes that the Fisher--Rao metric
is (up to a constant) the unique Riemannian metric on the probability
simplex that is invariant under sufficient statistics.  In
Section~\ref{sec:geometry} we show that for probability loading the
pullback of the quantum Fubini--Study metric satisfies
$\phi^* g_{FS} = \tfrac{1}{4}\, g_{\text{Fisher}}$, collapsing to
this classical geometry---a direct consequence of the vanishing Berry
connection.

\section{Feature Maps and Hypothesis Classes}
\label{sec:features}

Let $\mathcal{X} \subseteq \Real^d$ be a data domain and
$\Hil \cong \Complex^N$ with $N = 2^n$ for $n$ qubits.

\begin{definition}[Quantum feature map]
A quantum feature map is a measurable map
$\Phi : \mathcal{X} \to \mathcal{D}(\Hil)$, where $\mathcal{D}(\Hil)$
is the set of density operators on~$\Hil$.  We write
$\rho(x) = \Phi(x)$ and $\ket{\psi(x)}$ when $\rho(x)$ is pure.
\end{definition}

\begin{definition}[PQC hypothesis class]
A parameterised quantum circuit (PQC) model is specified by a feature
map~$\Phi$, a data-independent unitary family
$\{U(\theta) : \theta \in \Theta \subseteq \Real^m\}$, and an
observable $O \in \mathrm{Herm}(\Hil)$.  The hypothesis class is
\begin{equation}
  \mathcal{F}(\Phi, U, O)
    = \bigl\{x \mapsto
      \mathrm{tr}\bigl[O\,U(\theta)\,\rho(x)\,U(\theta)^\dag\bigr]
      : \theta \in \Theta\bigr\}.
\end{equation}
\end{definition}

\subsection{Canonical Encoding Strategies}
\label{sec:five}

We consider the following canonical feature maps. Throughout,
$x = (x_1, \ldots, x_d) \in \mathcal{X}$ and
$p(x) = (p_0(x), \ldots, p_{N-1}(x))$ is a probability vector.

\paragraph{Real-valued amplitude encoding (RAE).}
$\ket{\psi_{\mathrm{RAE}}(x)} = \sum_{i=0}^{N-1} c_i(x)\ket{i}$,
with $c_i(x) \in \Real$ and $\sum_i c_i(x)^2 = 1$.
This encompasses \emph{probability loading} (PL) as the non-negative
sub-case $c_i(x) = \sqrt{p_i(x)} \geq 0$.

\begin{remark}[On ``amplitude encoding'']
\label{rem:amplitude}
The term ``amplitude encoding'' in the QML literature sometimes
refers to the general family permitting complex amplitudes
($c_i \in \Complex$), including Approximate Complex Amplitude Encoding
(AAE).  Our Equivalence Theorem (Section~\ref{sec:equiv}) applies only
to the sub-family in which amplitudes are restricted to real values.
Complex-phase amplitude encodings can break the classical equivalence
established below.
\end{remark}

\paragraph{Probability loading (PL).}
$\ket{\psi_{\mathrm{PL}}(x)} = \sum_{i=0}^{N-1}
\sqrt{p_i(x)}\,\ket{i}$, with all amplitudes real and non-negative.
A special case of RAE; requires $n = \lceil\log_2 N\rceil$ qubits.

\paragraph{Angle encoding (AE).}
$\ket{\psi_{\mathrm{AE}}(x)} = \bigotimes_{j=1}^{n}
\bigl(\cos(x_j/2)\ket{0} + \sin(x_j/2)\ket{1}\bigr)$, encoding
$d = n$ features into single-qubit $R_y$ rotations.

\paragraph{Data re-uploading (RU).}
$\ket{\psi_{\mathrm{RU}}(x)} =
U_L(\theta_L)\,\Phi_0(x)\cdots U_1(\theta_1)\,\Phi_0(x)\ket{0}$,
where $\Phi_0$ is a base encoding repeated $L$ times
\cite{PerezSalinas2020}.

\paragraph{Sandwich encoding (SW).}
$\ket{\psi_{\mathrm{SW}}(x)} =
R_Y(x/2)\,e^{-i\lambda ZZ}\,R_Y(x/2)\ket{0^n}$, a symmetric
second-order Suzuki--Trotter decomposition in which the data
Hamiltonian is split equally on both sides of the entangling layer.
The non-commutativity
$[R_Y(x),\, e^{-i\lambda ZZ}] \neq 0$ ensures that the data actively
modulates quantum phases.

\paragraph{Hamiltonian encoding (HE).}
$\ket{\psi_{\mathrm{HE}}(x)} = e^{-iH(x)t}\ket{0^n}$, where
$H(x) = \sum_j x_j G_j + \sum_{j<k} x_j x_k G_{jk}$ with Pauli
generators $G_j, G_{jk}$.

\begin{remark}[Qubit cost]
RAE/PL encode $N = 2^n$ values into $n$ qubits---exponential
compression.  All other encodings require $n = d$ qubits for
$d$-dimensional input.  The compression advantage is real; however,
it comes at the cost of hypothesis class expressivity, as we now show.
\end{remark}

\section{The Equivalence Theorem}
\label{sec:equiv}

\subsection{Kernels for Real-Valued Amplitude Encoding}

\begin{definition}[Inner-product kernel on $\Real^N$]
\label{def:ipkernel}
For vectors $c, c' \in \Real^N$ with $\|c\|_2 = \|c'\|_2 = 1$, the
inner-product kernel is
\begin{equation}
  k_{\Real}(c, c') = \bigl(c^\top c'\bigr)^2
    = \bigl|\!\braket{\psi_c}{\psi_{c'}}\!\bigr|^2,
\end{equation}
a classical, positive-definite kernel computable in $O(N)$ time.
\end{definition}

\begin{definition}[Hellinger affinity kernel]
\label{def:hellinger}
For probability vectors $p, q$ over $[N]$,
\begin{equation}
  k_H(p, q)
    = \Bigl(\sum_{i=0}^{N-1} \sqrt{p_i\, q_i}\Bigr)^{\!2},
\end{equation}
the non-negative sub-case of $k_{\Real}$.
\end{definition}

\subsection{Main Result}

\begin{theorem}[Equivalence Theorem---Real-Valued Amplitude Encoding]
\label{thm:equiv}
Let $\ket{\psi_c(x)} = \sum_{i=0}^{N-1} c_i(x)\ket{i}$ be a
real-valued amplitude encoding with $c_i(x) \in \Real$ and
$\sum_i c_i(x)^2 = 1$.  For any data-independent unitary $U(\theta)$
and any observable $O \in \mathrm{Herm}(\Hil)$,
\begin{equation}
  f_\theta(x)
    = \bra{\psi_c(x)} U(\theta)^\dag\, O\, U(\theta) \ket{\psi_c(x)}
    = c(x)^\top M(\theta)\, c(x),
\end{equation}
where $M(\theta) = U(\theta)^\dag O\, U(\theta)$ and
$c(x) = (c_0(x), \ldots, c_{N-1}(x))^\top$.

The hypothesis class of any PQC built on real-valued amplitude
encoding is
\begin{equation}
  \mathcal{F}_{\mathrm{RAE}}
    \subseteq \bigl\{x \mapsto c(x)^\top A\, c(x)
      : A \in \mathrm{Sym}(\Real^N)\bigr\},
\end{equation}
which is the reproducing kernel Hilbert space (RKHS) of $k_{\Real}(x,x') = (c(x)^\top c(x'))^2$.
For probability loading, this reduces to the RKHS of the Hellinger
kernel~$k_H$.
\end{theorem}

\begin{proof}
Expanding the quantum expectation value directly in the computation basis:
\begin{equation}
  f_\theta(x)
    = \sum_{i,j} c_i(x)\, M_{ij}(\theta)\, c_j(x)
    = c(x)^\top M(\theta)\, c(x).
\end{equation}
Since the state coefficients are strictly real-valued ($c(x) \in \Real^N$), the 
imaginary components of the Hermitian matrix $M(\theta)$ cancel symmetrically, meaning 
only the real part contributes to the expectation:
\begin{equation}
  f_\theta(x) = c(x)^\top \mathrm{Re}(M(\theta))\, c(x).
\end{equation}

For any $A \in \operatorname{Sym}(\mathbb{R}^N)$, choosing
$O = U(\theta)\,A\,U(\theta)^\dagger \in \operatorname{Herm}(\mathcal{H})$
gives $\operatorname{Re}(M(\theta)) = \operatorname{Re}(A) = A$;
hence as $O$ ranges over $\operatorname{Herm}(\mathcal{H})$, the map
$O \mapsto \operatorname{Re}(U(\theta)^\dagger O\, U(\theta))$ is surjective
onto $\operatorname{Sym}(\mathbb{R}^N)$ for every fixed $\theta$.

The class of functions $x \mapsto c(x)^\top A\, c(x)$ for $A \in \operatorname{Sym}(\mathbb{R}^N)$
is exactly the RKHS of $k_R(x,x') = (c(x)^\top c(x'))^2$.
To see this, note that $k_R(x,x') = \langle c(x) \otimes c(x),\, c(x') \otimes c(x')\rangle$
embeds each point into the symmetric tensor product $\mathbb{R}^N \odot \mathbb{R}^N$;
functions of the form $c(x)^\top A\, c(x)$ are precisely the linear functionals
on this feature space~\cite{scholkopf2002learning}.
\end{proof}

\begin{corollary}[Classical realisability]
\label{cor:classical}
Every function $f \in \mathcal{F}_{\mathrm{RAE}}$ can be computed
classically in $O(N^2)$ time, where $N$ is the dimension of the
amplitude vector (i.e., $N = 2^n$ for $n$ qubits).  
This is
polynomial in the data dimension~$N$ but exponential in the number of
qubits---matching the encoding's own cost.  The hypothesis class of
any real-valued amplitude encoding PQC therefore provides no Type~B
quantum advantage: a classical kernel machine with the same data
access computes the same functions.
\end{corollary}

\begin{corollary}[Commutativity of the pulled-back algebra]
\label{cor:commutative}
The pulled-back observable algebra on $\mathcal{X}$---the algebra of
functions $\{f_\theta : \mathcal{X} \to \Real\}$ generated by varying
$\theta$ and $O$---is commutative, since these are ordinary
real-valued functions that multiply pointwise.

The operator algebra on $\Hil$ itself remains the full
$\mathrm{Herm}(\Hil)$; it is only its image on $\mathcal{X}$ via
real-valued encoding that collapses to a purely commutative subalgebra
of real quadratic forms.
\end{corollary}

\begin{remark}[Why data re-uploading escapes this bound]
\label{rem:reupload}
When data enters multiple layers, the state is
$U_L(\theta)\,\Phi_0(x) \cdots U_1(\theta)\,\Phi_0(x)\ket{0}$,
and $f(x)$ becomes a multilinear form in the amplitudes of degree
$2L$.  For $L > 1$ this strictly exceeds the hypothesis class of
Theorem~\ref{thm:equiv}.
\end{remark}

\begin{remark}[Why complex amplitudes escape this bound]
\label{rem:complex}
When $c_i(x) \in \Complex$, the expectation value depends on the
\emph{full} Hermitian matrix $M(\theta)$, not just its real part.
The imaginary components of $M$ contribute through cross-products of
phases, generating a kernel that is not classically realisable in
general.
\end{remark}

\section{Quantumness Measures for Encodings}
\label{sec:quantumness}


Theorem~\ref{thm:equiv} establishes a strict, binary verdict for 
idealized real-valued amplitude encodings. However, practical QML 
architectures rarely inhabit such clear-cut extremes; instead, they 
employ complex, hybridized feature maps where the boundaries of quantum 
advantage become obscured. To transform our theoretical insights into 
an operational engineering tool, we require quantitative diagnostics 
capable of measuring non-classicality on a continuous scale. 

In this section, we introduce a complementary suite of metrics designed 
to isolate and quantify the exact channels through which quantum 
advantage emerges or collapses. We analyze these channels across three 
distinct physical dimensions: 
\begin{enumerate}
    \item \textbf{Phase Space Distribution ($Q[\Phi]$ and $C[\Phi]$):} To diagnose the presence of data-dependent quantum interference and track the structural presence of complex phases.
    \item \textbf{Information-Theoretic Topology ($I[\Phi]$):} To evaluate the presence of inter-mode quantum correlations and multi-qubit entanglement over the data domain.
    \item \textbf{Information-Geometric Excess ($\Delta g$):} To formally quantify the gap between the quantum Fubini–Study geometry and the classical Fisher–Rao geometry induced by the encoding.
\end{enumerate}
Together, these diagnostics allow us to audit any arbitrary quantum feature map and predict its susceptibility to the Inverse Born Rule Fallacy.

\subsection{Wigner Negativity: $Q[\Phi]$}
\label{sec:wigner}

Quasiprobability representations generalise classical probability
distributions to quantum states.  Unlike true probabilities, they
can take negative values; these negative regions are a signature
of non-classicality with no classical probabilistic analogue.
Our goal is to track when an encoding injects sufficient phase complexity 
to force the state representation into these non-classical regimes.

For continuous-variable systems, the standard quasiprobability
representation is the Wigner function
$W_\rho : \Real^{2n} \to \Real$, defined via the Weyl transform.
States with $W_\rho \geq 0$ everywhere (mixtures of Gaussian states)
admit a classical probabilistic description; negative regions indicate
genuinely quantum features~\cite{Hudson1974}.

For finite-dimensional (qubit) systems, discrete quasiprobability
representations have been developed by Gibbons, Hoffman, and
Wootters~\cite{Gibbons2004} and by Gross~\cite{Gross2006}.  Gross
proved a discrete Hudson theorem: in odd prime dimension, a pure state
has a non-negative discrete Wigner function if and only if it is a
stabiliser state.  For multi-qubit systems (dimension $2^n$, not
prime), generalised frameworks exist \cite{Rundle2017}.  In all cases,
the qualitative distinction persists: states with only real amplitudes
in the computational basis have non-negative quasiprobability
representations, while states with data-dependent complex phases
generically develop negative regions.

\begin{definition}[Wigner negativity of an encoding]
\label{def:Q}
Let $W_{\rho(x)}$ denote the quasiprobability representation of the
encoded state $\rho(x) = \Phi(x)$ (continuous-variable Wigner function
or a discrete analogue as appropriate).  The Wigner negativity of the
encoding~$\Phi$ is
\begin{equation}
  Q[\Phi] = \mathbb{E}_{x \sim P_{\mathcal{X}}}\!\left[
    \int W^-_{\rho(x)}(\alpha)\,d\alpha
  \right],
\end{equation}
where $W^-(\alpha) = \max(0, -W_\rho(\alpha))$ is the negative part,
and the integration (or summation, in the discrete case) runs over the
full phase space.
\end{definition}



\begin{remark}[The necessity of phase complexity {$C[\Phi]$}]
\label{rem:Q0}
It is tempting to associate classicality directly with the vanishing of Wigner negativity ($Q=0$). However, this is geometrically imprecise. By the Hudson--Gross theorem~\cite{Gross2006}, the only pure states with non-negative discrete Wigner functions are stabiliser states. General real-valued superpositions---including probability loading $\ket{\psi} = \sum_i \sqrt{p_i}\ket{i}$---are generically non-stabiliser states and therefore exhibit $Q > 0$ despite lacking complex phases.

We therefore adopt the \emph{phase complexity}
\begin{equation}
  C[\Phi] = \mathbb{E}_{x}\!\left[
    \frac{\|\mathrm{Im}\,\psi(x)\|^2}{\|\psi(x)\|^2}
  \right]
\end{equation}
as the primary non-classicality diagnostic for encodings. Unlike $Q$, which conflates real-valued amplitude distributions with phase-driven interference, $C = 0$ if and only if all amplitudes are real. This measure is convention-free, efficiently computable, and---crucially---vanishes under exactly the same condition that forces the Berry connection to zero (Proposition~\ref{prop:geometric}).
\end{remark}

\subsection{Correlational Topology: Mode Mutual Information $I[\Phi]$}
\label{sec:mutual}

While phase complexity tracks local state properties, quantum advantage 
frequently depends on the non-local structure of the feature map---specifically, 
the generation of multi-qubit entanglement across the circuit register. 
To diagnose whether an encoding generates meaningful internal correlations, 
we measure its information-theoretic mode separation.

\begin{definition}[Mode-wise mutual information of an encoding]
Partition the $n$ qubits into individual modes $1, \ldots, n$ and let
$\rho_j(x) = \mathrm{tr}_{-j}[\rho(x)]$ be the reduced state of
mode~$j$.  The mode mutual information is
\begin{equation}
  I[\Phi] = \mathbb{E}_{x \sim P_{\mathcal{X}}}\!\left[
    S\!\Bigl(\bigotimes_{j=1}^n \rho_j(x)\Bigr)
    - S\bigl(\rho(x)\bigr)
  \right],
\end{equation}
where $S(\sigma) = -\mathrm{tr}(\sigma \log \sigma)$ is the von
Neumann entropy.
\end{definition}

$I[\Phi]$ tracks total inter-mode entanglement in the encoded state,
averaged over the data distribution. It vanishes for pure product states
and scales monotonically toward maximum values for highly entangled states.

For probability loading, the encoded state may possess non-zero $I[\Phi]$ if the loaded classical distribution $p(x)$ contains correlations between its variables, as these classical correlations manifest as mathematical entanglement across the qubit partition. However, because $C = 0$, this entanglement lacks phase coherence. It merely reflects classical statistical structure embedded in a quantum tensor product, rather than genuinely non-local quantum interference.

\subsection{Information-Geometric Excess and the Berry Connection}
\label{sec:geometry}

Ultimately, the physical consequences of phase complexity and correlation 
topology express themselves through the geometry of the state manifold. 
To map this explicitly, we compare the differential geometry of the quantum 
embedding against foundational boundaries in classical information theory.

For a pure-state feature map
$\Phi(x) = \ket{\psi(x)}\!\bra{\psi(x)}$ the pullback of the
Fubini--Study metric to $\mathcal{X}$ is
\begin{equation}
  (\phi^* g_{FS})_{\mu\nu}(x)
    = \mathrm{Re}\bigl[
        \braket{\partial_\mu \psi}{\partial_\nu \psi}
        - \braket{\partial_\mu \psi}{\psi}
          \braket{\psi}{\partial_\nu \psi}
      \bigr],
\end{equation}
where $\partial_\mu = \partial / \partial x^\mu$.  The second term is
the square of the Berry connection
$A_\mu(x) = i\braket{\psi}{\partial_\mu \psi}$.

The classical Fisher--Rao metric on the probability simplex is
$g^{\mathrm{Fisher}}_{\mu\nu}(x) = \sum_i p_i^{-1}\,
\partial_\mu p_i\, \partial_\nu p_i$.

\begin{definition}[Information-geometric excess]
\begin{equation}
  \Delta g_{\mu\nu}(x)
    = (\phi^* g_{FS})_{\mu\nu}(x)
      - \tfrac{1}{4}\, g^{\mathrm{Fisher}}_{\mu\nu}(x).
\end{equation}
\end{definition}

\begin{proposition}[Geometric mechanism of the equivalence]
\label{prop:geometric}
For real-valued amplitude encoding ($c_i(x) \in \Real$,
$\sum_i c_i^2 = 1$), the Berry connection vanishes identically:
\begin{equation}
  A_\mu(x) \equiv i\braket{\psi}{\partial_\mu \psi} = 0
  \qquad \text{for all } x \in \mathcal{X}.
\end{equation}
Consequently, the Fubini--Study pullback reduces to
$(\phi^* g_{FS})_{\mu\nu} =
\braket{\partial_\mu \psi}{\partial_\nu \psi}$, and for probability
loading this equals
$\tfrac{1}{4} g^{\mathrm{Fisher}}_{\mu\nu}$, giving $\Delta g = 0$.

It is the restriction to $\Real$---not the $L_2$ normalisation
requirement per se---that forces the model into the classical regime.
\end{proposition}

\begin{proof}
Differentiating $\braket{\psi}{\psi} = 1$ gives
$\mathrm{Re}\,\braket{\psi}{\partial_\mu \psi} = 0$.  When all
amplitudes are real, $\braket{\psi}{\partial_\mu \psi} =
\sum_i c_i\, \partial_\mu c_i \in \Real$, so
$\braket{\psi}{\partial_\mu \psi} = 0$ and $A_\mu = 0$.

For probability loading, $c_i = \sqrt{p_i}$ gives
$\partial_\mu c_i = \partial_\mu p_i / (2\sqrt{p_i})$, whence
\begin{equation}
  (\phi^* g_{FS})_{\mu\nu}
    = \sum_i (\partial_\mu c_i)(\partial_\nu c_i)
    = \sum_i \frac{\partial_\mu p_i\, \partial_\nu p_i}{4\, p_i}
    = \tfrac{1}{4}\, g^{\mathrm{Fisher}}_{\mu\nu}. 
\end{equation}
\end{proof}

\begin{remark}[Physical intuition of the Berry Connection]
\label{rem:berry_intuition}
The Berry connection admits a concrete geometric interpretation that
clarifies why its vanishing implies classicality.

Consider a family of quantum states $\ket{\psi(x)}$ parameterised by
data~$x$.  As $x$ changes, the state changes.  Part of this change is
visible to measurement: the probabilities
$p_i(x) = |c_i(x)|^2$ shift, and this shift is captured by the
Fisher--Rao metric.  But there can be an additional, invisible
change: the phases of the amplitudes rotate.  The Berry connection
$A_\mu = i\braket{\psi}{\partial_\mu \psi}$ measures exactly this
invisible rotation---the component of $\ket{\partial_\mu \psi}$ that
is parallel to $\ket{\psi}$ and therefore undetectable by any
single measurement of~$|\psi|^2$.

When $A_\mu = 0$, all changes in the state as $x$ varies are
visible to measurement: the probabilities $p_i(x)$ capture
everything, and the Fisher--Rao metric is a complete description of
the data geometry.  This is the classical regime.

When $A_\mu \neq 0$, the state carries information in its
phases that probabilities alone cannot reveal.  If $x$ traverses a
closed loop in data space, the state accumulates a Berry phase
$\gamma = \oint A_\mu\, dx^\mu$ that has no classical analogue---it is
a global, topological property of the encoding that cannot be
decomposed into local probability changes.  Quantum interference can
access this phase; classical statistics cannot.

For real-valued amplitudes, $\braket{\psi}{\partial_\mu \psi} =
\sum_i c_i\,\partial_\mu c_i \in \Real$, and the normalisation
constraint forces this to zero.  There is no room for invisible
rotation because there are no phases to rotate.  This is why the
restriction to $\Real$, not the normalisation per se, is the
mechanism that locks the model into the classical regime.
\end{remark}



\subsection{Encoding Comparison}

Table~\ref{tab:measures} compiles our diagnostic values across the six canonical 
feature maps detailed in Section~\ref{sec:features}, providing an immediate 
structural profile of each strategy.
\begin{table}[t]
\centering
\caption{Quantumness measures for encoding strategies.  $Q$ (Wigner
  negativity) is generically positive for all non-stabiliser states,
  including real-valued encodings, and therefore does not distinguish
  classical from quantum regimes in the sense of
  Theorem~\ref{thm:equiv}.  The operationally relevant diagnostics
  are $C$ (phase complexity) and $\Delta g$ (information-geometric
  excess), which vanish if and only if all amplitudes are real.}
\label{tab:measures}
\begin{tabular}{lccccl}
\hline
\textbf{Encoding}
  & $Q[\Phi]$ & $C[\Phi]$ & $I[\Phi]$ & $\Delta g$
  & \textbf{Hypothesis class} \\
\hline
Prob.\ loading (PL)
  & ${>}\,0$ & 0 & ${\geq}\,0$ & 0
  & Hellinger kernel \\
General RAE
  & Varies & 0 & ${\geq}\,0$ & 0
  & Inner-product kernel \\
Angle encoding (AE)
  & ${>}\,0$ & 0 & 0 & 0
  & Fourier (quadratic) \\
Re-uploading (RU)
  & ${>}\,0$ & 0$^*$ & ${>}\,0$ & 0$^*$
  & Multilinear ($2L$) \\
Sandwich (SW)
  & ${>}\,0$ & ${>}\,0$ & ${>}\,0$ & ${>}\,0$
  & Non-commutative \\
Hamiltonian (HE)
  & ${>}\,0$ & ${>}\,0$ & ${>}\,0$ & ${>}\,0$
  & Non-commutative \\
\hline
\multicolumn{6}{l}{\footnotesize
  $^*$For $R_y$-based injection as in our experiments;
  $C, \Delta g > 0$ if $\Phi_0$ introduces complex phases.}
\end{tabular}
\end{table}

\section{Numerical Experiments}
\label{sec:experiments}

\subsection{Design Rationale}

The experiments are designed to validate three claims, not to
benchmark encoding strategies against one another.

\emph{Claim~1 (Equivalence Theorem).}
Probability loading collapses to chance on classification tasks whose
decision boundaries lie outside the quadratic-form hypothesis class.

\emph{Claim~2 (Diagnostic validity).}
The encoding diagnostics $C$ (phase complexity), $|A_\mu|$ (Berry
connection magnitude), and $I[\Phi]$ (mode mutual information),
computed on the encoding alone, predict whether a given encoding can, in principle, reach beyond the classical regime.

\emph{Claim~3 (Two escape routes).}
Theorem~\ref{thm:equiv} has two independent premises---real-valued
amplitudes and data-independent ansatz.  Breaking either one
suffices to escape the classical bound: complex-phase encodings
(SW, HE) break the first; data re-uploading (RU) breaks the second.

We select four classification tasks with progressively harder decision
boundaries, each chosen because of what the Equivalence Theorem
predicts about probability loading.

\subsection{Task Selection}

\paragraph{Two-moons.}
Two interleaving half-circles in $\Real^2$.  The boundary is roughly
linear---the Hellinger kernel suffices for coarse separation.
\emph{Prediction}: PL performs adequately.

\paragraph{Concentric circles.}
Inner and outer circles in $\Real^2$.  The boundary is radial,
$x_1^2 + x_2^2 = r^2$, which is a quadratic form in $x$---but the
softmax preprocessing in PL distorts the geometry, and the resulting
Hellinger kernel does not preserve radial structure.
\emph{Prediction}: PL fails.

\paragraph{XOR.}
Four Gaussian clusters at $(\pm 1, \pm 1)$ with labels determined by
$\mathrm{sign}(x_1 \cdot x_2)$.  The decision boundary requires a
function that changes sign with the product of features.  Under
probability loading, all amplitudes $\sqrt{p_i} \geq 0$, so the
quadratic form $c(x)^\top A\, c(x)$ cannot encode the required sign
flip through the data.
\emph{Prediction}: PL fails.
\newpage
\paragraph{Parity (3-bit).}
Labels are $y = x_1 \oplus x_2 \oplus x_3$ (thresholded at zero).
This requires a three-body interaction---no pairwise function
suffices.  Even phase-active encodings require sufficient circuit depth
to express this function.
\emph{Prediction}: PL fails; shallow phase-active encodings may also
struggle.

\subsection{Setup}

All experiments use a pure NumPy state-vector simulator.  Each
encoding is followed by a hardware-efficient ansatz (alternating
$R_y$--$R_z$ rotations and CNOT ladders).  We use 2~qubits for
Two-moons, Circles, and XOR; 3~qubits for Parity.
Training uses Adam with finite-difference gradients over 300
iterations, batch size 32, learning rate 0.05.
Accuracy is reported on a held-out test set of 500~samples, averaged
over 5 random seeds.  Table~\ref{tab:config} reports the
configuration for each encoding.

\begin{table}[t]
\centering
\caption{Encoding configurations.  Parameters counted for $n = 3$
  qubits.  RU has the most parameters due to trainable rotations
  between data re-uploads.}
\label{tab:config}
\begin{tabular}{lccc}
\hline
\textbf{Encoding} & \textbf{Enc.\ layers} & \textbf{Ansatz depth}
  & \textbf{Total params} \\
\hline
Probability loading (PL)  & 1 & 3 & 18 \\
Angle encoding (AE)       & 1 & 3 & 18 \\
Data re-uploading (RU)    & 3 & 2 & 30 \\
Sandwich (SW)             & 3 & 2 & 15 \\
Hamiltonian (HE)          & 3 & 2 & 15 \\
\hline
\end{tabular}
\end{table}

\subsection{Encoding Diagnostics}
\label{sec:diagnostics}

Before training, we compute three diagnostics on the encoded states
(no ansatz), averaged over 100 data points drawn uniformly from each
task.

\emph{Phase complexity} $C[\Phi]$: the fraction of the state norm
carried by imaginary amplitudes,
$C = \mathbb{E}_x\bigl[\|\mathrm{Im}\,\psi(x)\|^2 /
\|\psi(x)\|^2\bigr]$.  This is a simple, convention-free proxy for
the presence of complex phases: $C = 0$ if and only if all amplitudes
are real.

\emph{Berry connection magnitude} $|A_\mu|$: the absolute value
$|\!\braket{\psi}{\partial_\mu \psi}\!|$ averaged over feature
dimensions and data points.  By Proposition~\ref{prop:geometric},
$|A_\mu| = 0$ for all real-valued encodings.

\emph{Mode mutual information} $I[\Phi]$: the total von Neumann
entanglement across qubit modes
(Definition in Section~\ref{sec:mutual}).

Table~\ref{tab:diagnostics} reports the results.

\begin{table}[t]
\centering
\caption{Encoding diagnostics, averaged over data points.  PL and AE
  are strictly classical ($C = 0$, $|A| = 0$).  RU uses only
  $R_y$ and CNOT gates, which are real-valued, so $C = 0$ despite
  multi-layer structure.  SW and HE produce genuinely complex states.}
\label{tab:diagnostics}
\begin{tabular}{lcccc}
\hline
\textbf{Encoding}
  & $C[\Phi]$ & $|A_\mu|$ & $I[\Phi]$
  & \textbf{Regime} \\
\hline
PL   & 0.000 & 0.000 & 0.15 & Classical \\
AE   & 0.000 & 0.000 & 0.00 & Classical (product) \\
RU   & 0.000 & 0.000 & 0.72 & Classical (entangled) \\
SW   & 0.26  & 0.92  & 0.69 & Quantum \\
HE   & 0.11  & 0.49  & 0.34 & Quantum \\
\hline
\end{tabular}
\end{table}
Two observations are immediate. First, PL, AE, and RU all satisfy $C = 0$ and $|A_\mu| = 0$: their encoded states have purely real amplitudes. This is because $R_y$ rotations produce real entries and CNOT preserves this property. Second, we observe a residual mode mutual information $I[\Phi] = 0.15$ for PL. This non-zero value arises strictly from the normalization constraints and internal correlations of the classical distribution $p(x)$ across the state vector, embedding classical correlations as qubit entanglement without introducing quantum phase complexity. In contrast, SW and HE produce genuinely complex states ($C > 0$, $|A_\mu| > 0$) accompanied by strong phase-coherent entanglement.

This means that the diagnostics $C$, $|A_\mu|$, and $I$ characterise
the \emph{encoding}, not the full model. A model can succeed despite
$C = 0$ if it escapes the Equivalence Theorem through a different
mechanism---specifically, through multi-layer data re-uploading, which
violates the theorem's data-independent ansatz assumption rather than
its real-amplitude assumption.
\begin{figure}[t]
\centering
\includegraphics[width=\textwidth]{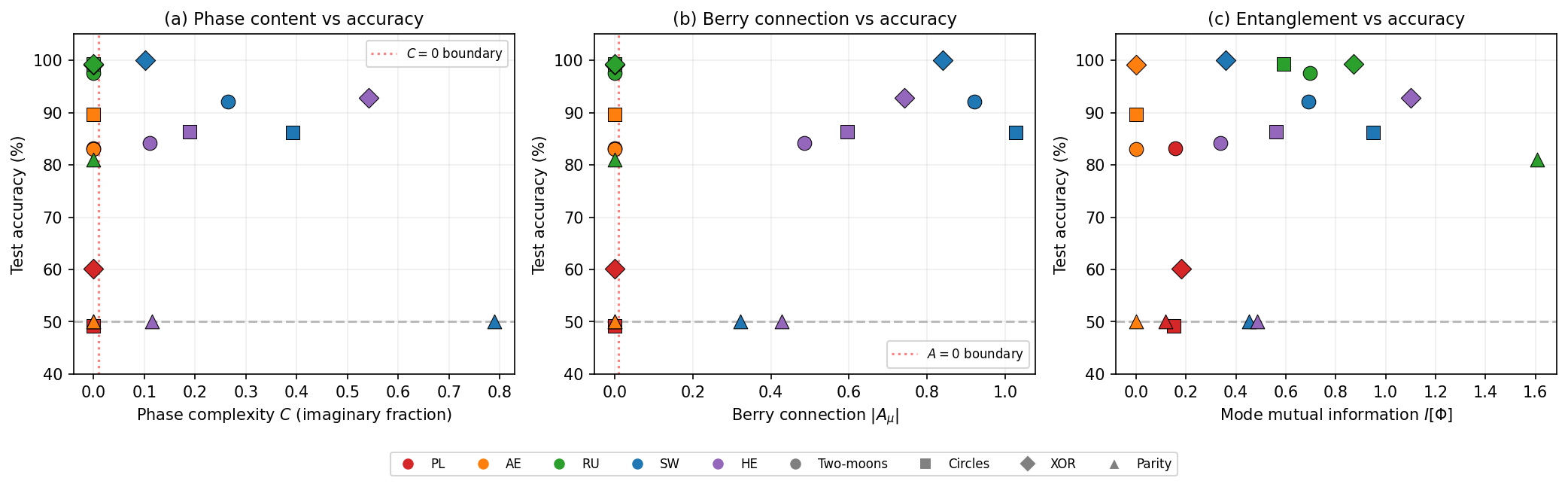}
\caption{Encoding diagnostics vs.\ classification accuracy.  Each
  point is one encoding--task pair (shape = task, colour = encoding).
  Left: phase complexity $C$ separates classical ($C = 0$: PL, AE,
  RU) from quantum ($C > 0$: SW, HE) encodings, but does not predict
  accuracy by itself.
  Centre: Berry connection $|A_\mu|$ shows the same separation.
  Right: mode mutual information $I$ is nonzero for entangling
  encodings (RU, SW, HE), including the classically-encoded RU.}
\label{fig:diagnostics}
\end{figure}

\subsection{Classification Accuracy}

Table~\ref{tab:accuracy} reports test accuracy.

\begin{table}[t]
\centering
\caption{Test accuracy (\%, mean $\pm$ std over 5 seeds).
  PL collapses on all nonlinear tasks.  AE and RU solve Parity
  (${\sim}81\%$) through deep ansatz circuits despite $C = 0$,
  while SW/HE fail despite quantum encoding diagnostics---indicating
  that circuit depth, not encoding quantumness alone, determines
  Parity performance.}
\label{tab:accuracy}
\begin{tabular}{lcccc}
\hline
\textbf{Encoding}
  & \textbf{Two-moons} & \textbf{Circles}
  & \textbf{XOR} & \textbf{Parity} \\
\hline
PL  & $85.0 \pm 1.4$ & $48.1 \pm 6.0$
    & $60.6 \pm 12.3$ & $48.2 \pm 2.5$ \\
AE  & $83.1 \pm 0.5$ & $90.1 \pm 3.1$
    & $99.1 \pm 0.4$ & $81.6 \pm 4.6$ \\
RU  & $97.4 \pm 0.8$ & $99.3 \pm 0.4$
    & $99.0 \pm 0.8$ & $80.8 \pm 4.5$ \\
SW  & $93.4 \pm 5.4$ & $87.6 \pm 2.4$
    & $100.0 \pm 0.0$ & $58.5 \pm 5.8$ \\
HE  & $82.1 \pm 2.5$ & $83.1 \pm 1.9$
    & $100.0 \pm 0.1$ & $51.2 \pm 0.8$ \\
\hline
\end{tabular}
\end{table}

The results confirm Claim~1 and refine Claims~2 and~3.

\emph{Claim~1}: PL collapses to near-chance on Circles
($48.1\%$), XOR ($60.6\%$), and Parity ($48.2\%$), consistent with
Theorem~\ref{thm:equiv}.  On Two-moons ($85.0\%$) it performs
adequately, as predicted for a roughly linearly separable task.

\emph{Claim~2}: The encoding diagnostics correctly separate
classical from quantum encodings: $C = 0$ and $|A_\mu| = 0$ for
PL, AE, and RU; $C > 0$ and $|A_\mu| > 0$ for SW and HE.
However, encoding diagnostics alone do not predict accuracy---the
ansatz depth and parameter count matter independently.

\emph{Claim~3 (Two escape routes, refined).}
The Parity task reveals a striking pattern: AE ($C = 0$, $I = 0$,
18~params, depth~3) achieves $81.6\%$, while SW ($C = 0.79$,
$I = 0.45$, 15~params, depth~2) achieves only $58.5\%$.  The
classically-encoded AE outperforms the quantum-encoded SW, because
AE's deeper ansatz (depth~3 with CNOT entanglement) compensates for
its encoding's lack of quantum content.  
\emph{Claim~3 (The trigonometric baseline and escape routes, refined).}
The experiments reveal a crucial distinction in how encodings interact with the Equivalence Theorem. 
AE ($C = 0$) operates entirely within the bounds of Theorem~\ref{thm:equiv}: its hypothesis class is strictly a classical quadratic form. However, it achieves strong performance on Circles ($90.1\%$) and XOR ($99.1\%$). This is not a violation of the theorem, but a consequence of its basis: AE maps data to trigonometric amplitudes ($\cos(x_j/2), \sin(x_j/2)$). The resulting quadratic form $c(x)^\top A\, c(x)$ generates a finite Fourier series, allowing it to easily resolve periodic and non-linear boundaries that strictly spatial probability loading ($\sqrt{p(x)}$) cannot. 

To genuinely escape the quadratic-form limitation of Theorem~\ref{thm:equiv} and access structurally deeper quantum hypothesis classes, an encoding must take one of two routes:
\begin{itemize}
\item \emph{Route~1: Complex phases.}
  SW and HE break the real-amplitude premise ($C > 0$). On XOR, where the task is two-dimensional and the parameter budget is sufficient, this yields perfect accuracy ($100\%$) with fewer parameters than classical alternatives.
\item \emph{Route~2: Data re-uploading.}
  RU breaks the data-independent ansatz premise. Through multi-layer data injection ($L = 3$), it generates high-degree polynomial features even with purely real gates, achieving robust performance ($97$--$99\%$) across spatial tasks where a single-layer real encoding fails.
\end{itemize}

The practical conclusion is that quantum encoding diagnostics
($C$, $|A_\mu|$, $I$) characterize the \emph{potential} quantum
content of the encoding, but accuracy depends on the interplay
between encoding, ansatz depth, and parameter count.  At equal depth
and parameters, quantum encodings (Route~1) outperform classical
ones—SW achieves $100\%$ on XOR with 15 parameters, while PL
achieves only $60.6\%$ with 18.  But sufficient depth (Route~2) can
compensate for classical encodings on many tasks.

\begin{figure}[t]
\centering
\includegraphics[width=\textwidth]{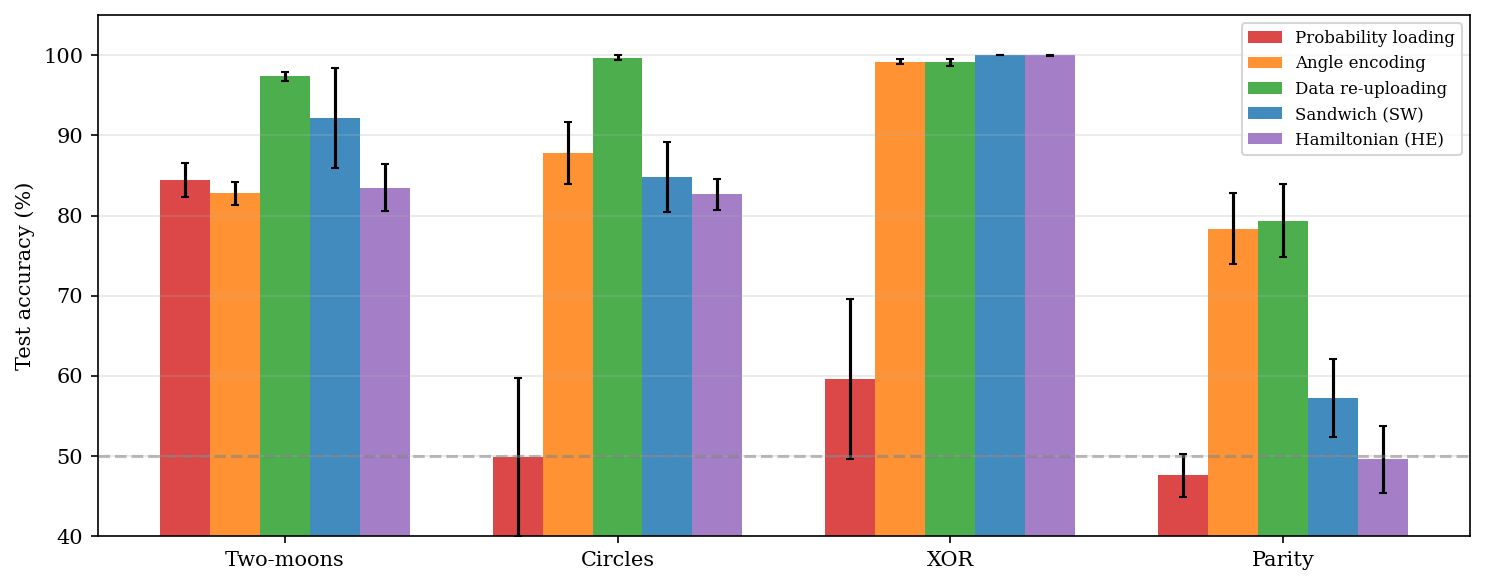}
\caption{Test accuracy across four tasks and five encodings.
  PL (red) collapses to near-chance on Circles, XOR, and Parity.
  On XOR, SW and HE achieve $100\%$ with fewer parameters than PL
  ($60.6\%$), demonstrating quantum encoding advantage at equal
  depth.  On Parity, AE and RU succeed through deep ansatz circuits
  despite $C = 0$.  Dashed line: chance level ($50\%$).}
\label{fig:accuracy}
\end{figure}

\section{Discussion}
\label{sec:discussion}
\paragraph{Encoding diagnostics vs.\ model performance.}
The phase complexity~$C$, Berry connection~$|A_\mu|$, and mutual
information~$I$ are properties of the \emph{encoding} $\Phi$,
computed before any ansatz is applied. The Parity experiment
illustrates that encoding diagnostics alone do not predict
accuracy. AE ($C = 0$) outperforms SW ($C = 0.79$) despite lacking phase complexity, not because it breaks the Equivalence Theorem, but because its trigonometric basis ($\cos, \sin$) naturally expresses finite Fourier series, effectively resolving nonlinear periodic boundaries that spatial probability loading cannot. Therefore, $C > 0$ indicates quantum resources in the encoding; $C = 0$ indicates that the model is bounded to a classical quadratic form, and its success relies purely on the choice of basis (AE) or circuit architecture (re-uploading).

\paragraph{Potential vs.\ realised advantage.}
A subtle but important distinction must be drawn between the
\emph{quantum potential} of an encoding and the \emph{realised
quantum advantage} on a given task.  Phase complexity $C$ and Berry
connection $|A_\mu|$ characterise the potential: sandwich encoding
produces $C > 0$ for every non-trivial input $x \neq 0$, because the
$e^{-i\lambda Z_j Z_k}$ interaction generates complex phases
regardless of the data.  This is a property of the map~$\Phi$, not
of the dataset.

However, whether these quantum resources translate into advantage
depends on two further conditions.  First, the \emph{data must have
structure that the encoding's quantum features can resolve}.
On XOR, the parity of the $ZZ$ eigenvalues mirrors the sign structure
$\mathrm{sign}(x_1 \cdot x_2)$ that defines the task---the encoding
and the problem are geometrically aligned, and SW achieves $100\%$.
On Two-moons, the decision boundary is nearly linear, and the
additional phase structure contributes little beyond what the
Hellinger kernel already captures.
Second, the \emph{circuit must have sufficient depth and parameters}
to channel the encoding's quantum resources toward the task.
On Parity, SW has $C = 0.79$ and $I = 0.45$---the encoding provides
quantum resources---but with only three trainable coupling parameters,
the model cannot express the required three-body interaction.

The Equivalence Theorem provides a one-sided guarantee: $C = 0$ and
$|A_\mu| = 0$ imply that no quantum advantage is possible, regardless
of task or depth.  The converse does not hold: $C > 0$ is necessary
but not sufficient.  Quantum advantage is the product of encoding
potential, data--encoding alignment, and circuit capacity.  The theorem
sets one factor to zero for real-valued amplitude encoding; the
experiments show that the remaining factors matter independently.

\paragraph{The two escape routes are complementary.}
Route~1 (complex phases) and Route~2 (depth/re-uploading) are
independent and can be combined.  The experiments suggest that at
equal parameter budget, Route~1 is more efficient on tasks where the
relevant features involve phase interference (XOR: SW achieves
$100\%$ with 15~parameters vs.\ PL's $60.6\%$ with 18), while
Route~2 is more robust across diverse tasks (RU achieves
$97$--$99\%$ on Two-moons, Circles, and XOR with 30~parameters).
On Parity, neither route alone suffices with the tested
configurations; combining both (e.g., sandwich encoding with
data re-uploading) is a natural direction for future work.

\paragraph{State preparation overhead.}
Our result is independent of state preparation complexity: even with
an efficient preparation oracle, the hypothesis class of real-valued
amplitude encoding remains classically realisable.  The well-known
$O(N)$ circuit depth required for general amplitude
preparation~\cite{SchuldPetruccione2021} is an additional, orthogonal
limitation.

\paragraph{Qubit cost vs.\ expressivity.}
PL encodes $N = 2^n$ values into $n$ qubits---exponential
compression.  Phase-active encodings require $n = d$ qubits.
Whether compression outweighs the hypothesis class restriction
depends on the application: for tasks well-modelled by inner-product
similarity (e.g., document retrieval), PL is appropriate; for tasks
requiring nonlinear feature interactions, phase-active encodings or
multi-layer re-uploading are necessary.

\paragraph{The geometric boundary.}
Proposition~\ref{prop:geometric} locates the boundary between
classical and quantum regimes at the Berry connection: $|A_\mu| = 0$
for real-valued encodings, $|A_\mu| > 0$ for complex.  The
experiments confirm that $|A_\mu|$ is a readily computable diagnostic:
it separates PL/AE/RU ($|A_\mu| = 0$) from SW/HE ($|A_\mu| > 0$)
cleanly, and correlates with the presence of genuinely quantum
interference in the encoded state.

\paragraph{Limitations.}
Our analysis assumes a data-independent ansatz $U(\theta)$.  Adaptive
architectures where $U$ depends on $x$ fall outside our framework.
The diagnostics $C$, $|A_\mu|$, and $I$ require access to the full
state vector, which scales exponentially; for large circuits, sampling
approximations are necessary.
We do not address trainability: phase-active encodings may be subject
to barren plateaus~\cite{McClean2018}, though quantum-aware
optimisers~\cite{Stokes2020} can partially mitigate this.
Our results characterize the hypothesis class at the algorithmic level
and are therefore independent of hardware noise; validation on
noisy intermediate-scale quantum (NISQ) devices is
left for future work.


\section{Conclusion}
\label{sec:conclusion}

We have shown that real-valued amplitude encoding---of which
probability loading is the non-negative sub-case---provides no
Type~B quantum advantage: its hypothesis class is equivalent to a
classical inner-product kernel, regardless of ansatz depth or
entanglement (Theorem~\ref{thm:equiv}).

The geometric mechanism is precise: restricting amplitudes to~$\Real$
forces the Berry connection to vanish, collapsing the Fubini
–Study metric to the classical Fisher–Rao metric ($\Delta g = 0$).  
The critical boundary is therefore not the $L_2$ normalization, but
the restriction to real values.

The Equivalence Theorem bounds encodings to classical quadratic forms in their respective amplitude basis. Angle Encoding operates strictly within this bound, but utilizes a trigonometric basis to efficiently resolve periodic boundaries. To structurally escape the quadratic-form limitation entirely, models must exploit one of two routes. 
\emph{Route~1} (complex-phase encodings such as sandwich or Hamiltonian encoding) breaks the real-amplitude premise: data-dependent phases generate Berry connection $|A_\mu| > 0$, yielding perfect accuracy on XOR with fewer parameters than classical alternatives. \emph{Route~2} (deep ansatz circuits or data re-uploading) breaks the data-independent ansatz premise: multi-layer data injection yields high-degree polynomial features even with purely real gates, achieving strong performance on Parity where shallow encodings fail.
We introduced three encoding diagnostics—phase complexity $C$,
Berry connection magnitude $|A_\mu|$, and mode mutual information
$I[\Phi]$—and showed that they characterize the encoding's quantum
content independently of the ansatz.  Models with $C > 0$ exploit
genuinely quantum resources; models with $C = 0$ must rely on
classical mechanisms (depth, re-uploading) for expressivity.

We conclude that the choice of feature map is the primary design
decision in QML.  Encoding design should be guided by explicit
measures of quantum content rather than by computational convenience,
and the distinction between quantum resources in the encoding and
classical resources in the architecture should be made explicit in
any claim of quantum advantage.

\end{document}